\documentclass[journal,twoside,web]{ieeecolor}
\usepackage{lcsys}
\usepackage{cite}
\usepackage{amsmath,amssymb,amsfonts}
\usepackage{algorithmic}
\usepackage{graphicx}
\usepackage{textcomp}
\def\BibTeX{{\rm B\kern-.05em{\sc i\kern-.025em b}\kern-.08em
    T\kern-.1667em\lower.7ex\hbox{E}\kern-.125emX}}
\usepackage{array,multicol,multirow,booktabs,longtable,tabularx}
\usepackage{subfig}
\usepackage{amssymb,mathtools,dsfont,bbm,cuted,stmaryrd}
\usepackage{xspace,algorithmic,algorithm}

\usepackage{tikz}
\usetikzlibrary{math,arrows.meta,spy,automata,intersections,decorations,arrows,decorations.pathmorphing,decorations.markings,decorations.footprints,fadings,calc,trees,mindmap,shadows,decorations.text,patterns,positioning,shapes,matrix,fit,3d,plotmarks,svg.path}

\usepackage{pgfplots}

\usepgfplotslibrary{patchplots}

\usepackage{tikz-3dplot}
\usepackage{tkz-euclide}

\usepackage{placeins}
\usepackage{cuted}

\let\labelindent\relax
\usepackage{enumitem}

\newcommand{\be}{\begin{equation}}
\newcommand{\ee}{\end{equation}}

\def\bal#1\eal{\begin{align}#1\end{align}}
\def\baln#1\ealn{\begin{align*}#1\end{align*}}

\newcommand{\ben}{\begin{equation*}}
\newcommand{\een}{\end{equation*}}

\newcommand{\re}[1]{\mathbb{R}^{#1}}%

\newcommand{\bbm}{\begin{bmatrix}}
\newcommand{\ebm}{\end{bmatrix}}

\newcommand{\bBm}{\begin{Bmatrix}}
\newcommand{\eBm}{\end{Bmatrix}}

\newcommand{\bvm}{\begin{vmatrix}}
\newcommand{\evm}{\end{vmatrix}}

\newcommand{\bVm}{\begin{Vmatrix}}
\newcommand{\eVm}{\end{Vmatrix}}

\newcommand{\bpm}{\begin{pmatrix}}
\newcommand{\epm}{\end{pmatrix}}

\newcommand{\bnm}{\begin{matrix}}
\newcommand{\enm}{\end{matrix}}

\newcommand{\bi}{\begin{itemize}}
\newcommand{\ei}{\end{itemize}}

\newcommand{\bse}{\begin{subequations}}
\newcommand{\ese}{\end{subequations}}

\newcommand{\secref}[1]{Section~\ref{#1}\xspace}

\newcommand{\figref}[1]{Fig.~\ref{#1}\xspace}

\newenvironment{proof-sketch}{\noindent{ \textit{Sketch of Proof}:}\hspace*{0.5em}}

\newtheorem{prop}{Proposition}

\newtheorem{lem}{Lemma}

\newtheorem{rem}{Remark}

\newcommand{\eor}{\ensuremath{\hfill\blacklozenge}}

\title{On maximal positive invariant set computation for rank-deficient linear systems}
\author{Bogdan Gheorghe, Daniel Ioan, Cristian Flutur\\ Ionela Prodan, and Florin Stoican
\thanks{Bogdan Gheorghe, Daniel Ioan, Cristian Flutur, and Florin Stoican are with Faculty of Automation Control and Computer Science, National University of Science and Technology Politehnica Bucharest, Romania {\tt\small \{bogdan.gheorghe1807, daniel\char`_mihail.ioan, cristian.flutur, florin.stoican\}@upb.ro}}%
\thanks{Ionela Prodan is with the Univ. Grenoble Alpes, Grenoble INP$^\dagger$, LCIS, F-26000, Valence, France, $^\dagger$ Institute of Engineering and Management Univ. Grenoble Alpes.
        {\tt\small ionela.prodan@lcis.grenoble-inp.fr}}%
}

\usepackage{graphicx}      

\usepackage{amsmath,amssymb,mathtools,stmaryrd}
\usepackage{xcolor}

\usepackage[algo2e, linesnumbered, algoruled]{algorithm2e}

\usepackage{hyperref}

\usepackage{svg}

\pgfplotsset{compat=1.18}

\begin{document}

\pagestyle{empty} 
\maketitle
\thispagestyle{empty}

\begin{abstract}
The maximal positively invariant (MPI) set is obtained through a backward reachability procedure involving the iterative computation and intersection of predecessor sets under state and input constraints. 

However, standard static feedback synthesis may place some of the closed-loop eigenvalues at zero, leading to rank-deficient dynamics. This affects the MPI computation by inducing projections onto lower-dimensional subspaces during intermediate steps. By exploiting the Schur decomposition, we explicitly address this singular case and propose a robust algorithm that computes the MPI set in both polyhedral and constrained-zonotope representations.
\end{abstract}



\section{Introduction}
\label{sec:intro}

The computation reachable sets is a fundamental problem at the intersection of control theory, algebraic geometry, and convex optimization, serving to define the exact theoretical bounds for safe system operation \cite{ossareh_complexity_2024,marpi2024lcss}. As a cornerstone of modern constrained control, reachability theory provides the mathematical guarantees necessary for recursive feasibility\cite{chen1998quasi}  and stability in Model Predictive Control (MPC)\cite{rakovic2018handbook} applications across the aerospace, chemical, and robotics sectors. 

A prime example of reachable set is the Maximal Positive Invariant (MPI) set which denotes the largest positive invariant set which also respects feasibility constraints. It is computed via a set recurrence relation involving the iterative calculation and intersection of predecessor sets subject to state and input constraints. Its primary use is a terminal set in Model Predictive Control\cite{rakovic2018handbook}. MPC operates by iteratively solving a finite-horizon optimization problem at each sampling interval, applying only the initial control action before receding the prediction horizon. A primary vulnerability is the potential loss of recursive feasibility: an optimal state sequence calculated at time $k$ may drive the plant to a state at $k+1$ from which future constraint violations become unavoidable. To prevent this and ensure closed-loop stability and constraint admissibility beyond the prediction horizon, the MPC formulation typically forces the terminal state into an MPI set where the control action switches to a fixed terminal feedback law.


The feedback gain $K$ for this terminal law is usually computed via the resolution of a discrete-time algebraic Riccati equation. This unconstrained optimal design frequently results in placing one or more closed-loop eigenvalues at the origin ($0$). While mathematically optimal, this renders the resulting closed-loop state transition matrix rank-deficient, which presents significant geometric and computational challenges for calculating the required terminal MPI set. 
While the MPI remains full-dimensional it is the result of computations where the singular closed-loop matrix maps the system dynamics into a lower-dimensional subspace, fundamentally altering the structural properties of the intermediate sets generated during this recurrence sequence.

While established methodologies can compute these invariant sets for standard polyhedral constraints, current literature lacks efficient solutions for the severe computational burdens introduced by the set recurrence at high dimensions. To mitigate these, our recent work \cite{10552328} has leveraged constrained zonotopes. By adding algebraic equality constraints, these structures retain efficient linear operations while remaining closed under arbitrary intersections. This allows the exact MPI set recurrence to be executed entirely within the constrained zonotope framework. 

Our idea is straightforward but, to the best of our knowledge, was never treated in the literature:
\begin{enumerate}[label=\roman*)]
    \item exploit the Schur factorization of the closed-loop matrix to identify the sub-space corresponding to non-zero eigenvalues and compute the reduced MPI set;
    \item lift the reduced MPI back into the original space and intersect it with ``initial constraints'' to output the MPI set w.r.t. the original dynamics and feasible set. 
\end{enumerate}


The remainder of this paper is organized as follows. \secref{sec:prel} provides the preliminary mathematical background on MPI sets and their geometrical representations. \secref{sec:justication} introduces the problem and provides theoretical justification for rank-deficient closed loop systems. \secref{sec:main_idea} details the MPI-computing procedure in both the polyhedral and constrained zonotopic representations. \secref{sec:ill_example} provides illustrative examples, while \secref{sec:conclusions} draws the conclusions.

\noindent \textit{Notations}: For a $x \in \mathbb R^{n}$, its infinity norm as $\|x\|_{\infty} := \max(|x_1|, |x_2|, \ldots, |x_n|)$. For $A \in \mathbb R^{n\times n}$, $\text{eig}(A)$ is  the set of all its eigenvalues. $I_n\in \mathbb R^{n\times n}$ is the identity matrix, and $\mathbf{1}_n$ is the column vector of $n$ values of one 

\clearpage

\section{Preliminaries}
\label{sec:prel}

Consider the linear time-invariant system
\begin{equation}
\label{eq:dynamics_open_loop}
x_{k+1} = A x_k + B u_k,
\end{equation}
where $x_k, x_{k+1} \in \mathbb{R}^n$ denote the current and successor states; $u_k\in \mathbb{R}^m$ represents the input; $A \in \mathbb{R}^{n \times n}$, $B \in \mathbb{R}^{n \times m}$ are the state and input matrices. The sets $\mathcal{X} \subset \mathbb{R}^n$ and $\mathcal{U} \subset \mathbb{R}^m$ impose constraints on the system state and input, respectively.


Assuming the pair $(A,B)$ stabilizable, there exists a static gain $K \in \mathbb{R}^{m \times n}$ defining the feedback law $u_k = Kx_k$, which leads to the closed-loop matrix $A_\circ := A + BK$ that stabilizes the system. The resulting closed-loop dynamics of \eqref{eq:dynamics_open_loop} become
\begin{equation}
\label{eq:LTI_dynamics_closed}
x_{k+1} = A_\circ x_k .
\end{equation}

Furthermore, the open-loop constraints $x_k \in \mathcal{X}$ and $u_k \in \mathcal{U}$ imply the closed-loop state constraint $x_k \in \overline{\mathcal{X}}$ with 
\begin{equation}
\label{eq:barx}
\overline{\mathcal{X}} = \mathcal{X} \cap \{x:\, Kx \in \mathcal{U}\}.
\end{equation}

A popular tool in the arsenal of set-based methods is the so-called maximal positively invariant (MPI) set \cite{gilbert1991linear,kolmanovsky_theory_nodate}. It is often used to define the terminal set in model predictive control architectures \cite{rakovic2018handbook}, ensuring recursive feasibility since it is, by definition, safe (due to positive invariance) and non-conservative (the ``maximal'' property). For our purposes, it suffices to recall the standard set recurrence that computes it:
\begin{equation}
\label{eq:mpi_rec_standard}
\Omega_0=\overline{\mathcal X}, 
\qquad 
\Omega_{k+1}=A_\circ^{-1}\Omega_k \cap \overline{\mathcal X}.
\end{equation}
Under mild assumptions \cite[Thm.~3]{rakovic2022implicit}, the recurrence \eqref{eq:mpi_rec_standard} is guaranteed to terminate at a fixed point $\Phi_x := \Omega_{\bar k}=\Omega_{\bar k+1}$ for some finite index $\bar k$. Thus, an equivalent form to \eqref{eq:mpi_rec_standard} is
\begin{equation}
    \label{eq:mpi_rec_standard_2}
    \Phi_x=\bigcap\limits_{k=0}^{\bar k} \bigl\{x:\: A_\circ^kx \in \overline{\mathcal X}\bigr\}.
\end{equation}
Further insights and variations of the basic MPI construction can be found in \cite{blanchini2008set,houska_polyhedral_2025} and subsequent works.

\subsection*{Set descriptions}

While the recurrence \eqref{eq:mpi_rec_standard} is agnostic with respect to the nature of $\overline{\mathcal X}$, particular choices do influence the subsequent computations. A traditional option is the polyhedral representation \cite{ziegler2012lectures}. Writing $\mathcal X=\{F_Xx\leq \theta_X\}$, $ \mathcal U=\{F_Uu\leq \theta_U\}$ in their half-space form allows to express \eqref{eq:barx} as
\begin{equation}
\label{eq:hrep}
\overline{\mathcal X}:=\{F x \leq \theta\}=\biggl\{\bbm F_X\\ F_UK\ebm x\leq \bbm \theta_X\\\theta_U\ebm\biggr\},
\end{equation}
allows the recurrence \eqref{eq:mpi_rec_standard} to be written explicitly as
\begin{equation}
\label{eq:mpi-recurrence-polyhedral}
    \Phi_x=\bigcap\limits_{k=0}^{\bar k} \bigl\{x:\: FA_\circ^kx \leq \theta \bigr\}.
\end{equation}

An alternative representation is that of the constrained zonotope (CZ), first introduced in \cite{scott2016constrained} and increasingly used for set manipulations \cite{raghuraman2022set}. We call the compact set $\mathcal{CZ} \subset \re{n}$ a constrained zonotope if there exists $c \in \mathbb R^{n}, G \in \mathbb R^{n \times D}, F \in \mathbb \re^{p\times D}, \theta \in \mathbb R^{p}$ such that
\begin{equation}
\label{eq:con_poly_zono}
    \mkern-8mu\mathcal{CZ} = \bigl\langle c, G, F, \theta\bigr\rangle = \big \{x = c + G\lambda,\: F\lambda=\theta,\: \|\lambda\|_\infty \leq 1\bigr\}.
\end{equation}




Noteworthy, CZs are closed under set intersection \cite{scott2016constrained}:
\begin{equation}
\label{eq:closedIntersection}
    \mkern-10mu \mathcal{CZ}_1 \cap \mathcal{CZ}_2 = \left\langle \mkern-4mu c_1, \bbm G_1 & 0 \ebm,\mkern-4mu  \bbm F_1 & \hphantom{-}0 \\ 0 & \hphantom{-}F_2 \\ G_1 & -G_2 \ebm, \bbm \theta_1 \\ \theta_2 \\ c_2 - c_1\ebm \right\rangle.
\end{equation}
Considering the state and input constraint sets as constrained zonotopes\footnote{In many practical cases the constraint sets are hyperrectangles, and therefore equality constraints are not required in their definition.}
$\mathcal X=\langle c_X, G_X, [\;], [\;]\rangle, \quad \mathcal U=\langle c_U, G_U, [\;], [\;]\rangle$, allows, via \eqref{eq:closedIntersection}, the set \eqref{eq:barx} to be expressed\footnote{We abuse the notation and re-use $F,\theta$ for the equality part of the CZ description.} as a CZ:
\begin{multline}
\label{eq:x_constr_zon}
\overline{\mathcal X} = \left\langle c, G, F, \theta\right\rangle =\bigl\langle c, \bbm G& 0\ebm,\\ \bbm KG& -G_U\ebm, c_U-Kc\bigr\rangle.
\end{multline}
Based on \cite{10552328}, all sets $\Omega_k = \left\langle c_k, G_k, F_k, \theta_k\right\rangle$ in the recurrence \eqref{eq:mpi_rec_standard} are also CZs and can be computed recursively as
\begin{multline}
\label{eq:mpi_constr_zon}
\Omega_{k+1} =\Bigl\langle A_\circ^{-1}c_k, \bbm A_\circ^{-1}G_k& 0\ebm,\\
    \bbm F_k&0\\ 0& F\\ A_\circ^{-1}G_k& -G\ebm,  \bbm\theta_k\\\theta\\c-A_\circ^{-1}c_k\ebm\Bigr\rangle.
\end{multline}

\vspace{-1.5em}
\subsection*{Illustrative example}
We take the discrete invariant linear dynamics defined by: \newline
$$A = \bbm 1.38 & 0.76 \\ 0.16 & 1.87 \ebm,\: B = \bbm 1 \\ 1\ebm$$
with input and state constraints given by $\mathcal X = \{\|x\|_{\infty} \leq 1\}$, $\mathcal U=\{\|u\|_{\infty} \leq 1\}$. Typically, closing the loop can be done either by solving the algebraic Riccati equation or by doing explicit pole placement. Each returns a stabilizing static gain $K$. Computing the MPI for dynamics \eqref{eq:LTI_dynamics_closed} and feasible set \eqref{eq:barx} for each gain $K$ results in the iterations illustrated in Figure~\ref{fig:example_LTI_MPI}. 
\begin{figure}[!ht]
\centering
\subfloat[MPI set for Ricatti]{\label{fig:example_LTI_idare}\includegraphics[width=0.5\columnwidth]{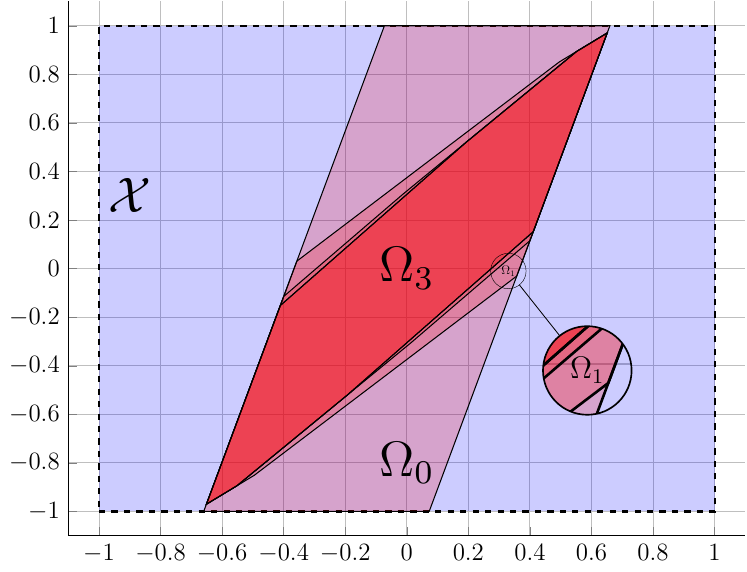}}\hfill
\subfloat[MPI set for pole placement]{\label{fig:example_LTI_place}\includegraphics[width=0.5\columnwidth]{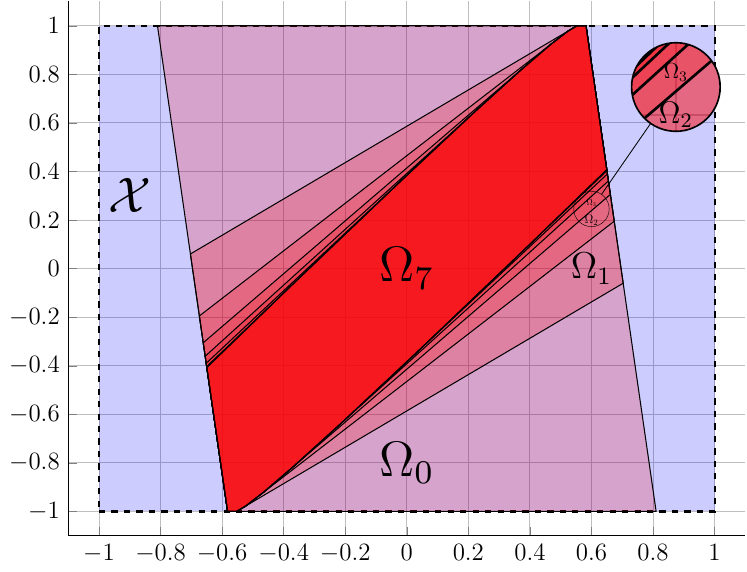}}
\caption{Example for computing the MPI set.}
\label{fig:example_LTI_MPI}
\end{figure}
Table~\ref{tab:MPI_method_comparison} shows that, even for this proof-of-concept example, significant differences appear in closed-loop in set shape and the value of $\bar k$, the index where the recurrence stops.
\begin{table}[!ht]
\centering
\renewcommand{\arraystretch}{1.35}
\setlength\tabcolsep{3pt}

\begin{tabular}{c|c|c|c}
Method & Poles & Gain ($K$) & Nr. of iter. ($k$)\\\hline
Ricatti & $\bbm 0.48 & 0.83\ebm$ & $\bbm 2.73 & -0.80\ebm$ & 3\\\hline
Placement & $\bbm 0.95 & 0.7\ebm$ & $\bbm 1.43 & 0.16\ebm$ & 7
\end{tabular}

\caption{Comparison of Ricatti and pole placement.}
\label{tab:MPI_method_comparison}
\setlength\tabcolsep{6pt}
\end{table}


\clearpage

\section{Problem formulation and justification}
\label{sec:justication}


The state matrix of the closed-loop dynamics \eqref{eq:LTI_dynamics_closed} may possess multiple eigenvalues at $\lambda = 0$ (equivalently, $0 \in \mathrm{eig}(A_\circ)$), possibly with different partial multiplicities. Such situations are not exceptional; rather, they naturally arise in a variety of standard control synthesis procedures.
%

One such situation is optimal control, where the closed loop system may get poles placed anywhere in the unit disc, provided the stabilizing solution is computed \cite{ionescu1997general}. We recall that the stabilizing solution of the Riccati equation is computed as the disconjugated invariant subspace of the Hamiltonian matrix, therefore the closed loop gain matrix $K$ may place the eigenvalues of $A+BK$ anywhere in the stable region, so it may even place more than one eigenvalue at $\lambda=0$. Another situation of interest may result from a dead-beat control strategy \cite{van1984deadbeat}, where all closed loop eigenvalues are placed at $\lambda=0$.

\subsection*{Numerical example}

For illustration and later use in the paper, we consider the following numerical example, where the gain matrix $K$ assigns parts of the closed loop spectrum to $\lambda=0$. Consider the controllable pair $(A,B)$ and the feedback gain matrix $K$: 

\[
{\small
A = \begin{bmatrix}
-14.85 & -5.20 & -14.75 & -11.90 & -20.10 & -14.55 \\
-\hphantom{0}8.85 & \hphantom{-}0.10 & -12.95 & -\hphantom{0}9.20 & -10.20 & -13.15 \\
\hphantom{-0}9.90 & \hphantom{-}6.60 & \hphantom{-}10.30 & \hphantom{-0}6.80 & \hphantom{-}13.80 & \hphantom{-}10.10 \\
-14.95 & -7.50 & -13.85 & -10.20 & -21.00 & -13.65 \\
-18.40 & -5.70 & -26.40 & -17.70 & -23.10 & -26.40 \\
-12.35 & -3.80 & -21.85 & -13.30 & -14.90 & -21.85
\end{bmatrix}
}
\]
$$
B =
\begin{bmatrix}
1 & 4 \\
3 & 4 \\
0 & 0 \\
0 & 2 \\
4 & 4 \\
4 & 2
\end{bmatrix},
\quad K =
\begin{bmatrix}
1 & 0 & 4 & 2 & 1 & 4 \\
3 & 1 & 2 & 2 & 4 & 2
\end{bmatrix}.
$$
For the resulting closed loop matrix $A_\circ=A+BK$, we compute the Real Schur form matrix ($A_\circ = USU^\top$). Matrix
\[
S = \left[\begin{tabular}{ccc|cc|c}
0.70 & -0.87 & 0.49 & -0.31 & \hphantom{-}35.16 & -13.00 \\
\hphantom{0.0}0 & \hphantom{-}0.50 & 0.12 & -0.23 & \hphantom{-0}2.27 & \hphantom{0}-0.85 \\
\hphantom{0.0}0 & \hphantom{-0.0}0 & 0.20 & -0.40 & \hphantom{0}-3.46 & \hphantom{-0}2.33 \\
\hline
\hphantom{0.0}0 & \hphantom{-0.0}0 & \hphantom{0.0}0 & -0.00 & \hphantom{0}-2.14 & \hphantom{-0}1.25 \\
\hphantom{0.0}0 & \hphantom{-0.0}0 & \hphantom{0.0}0 & \hphantom{-}0.00 & \hphantom{0}-0.00 & \hphantom{0}-0.00 \\
\hline
\hphantom{0.0}0 & \hphantom{-0.0}0 & \hphantom{0.0}0 & \hphantom{-0.0}0 & \hphantom{-00.0}0 & \hphantom{0}-0.00
\end{tabular}\right]
\]
has its eigenvalues at $\{0,0,0,0.2,0.5,0.7\}$. Furthermore, the eigenvalues at $\lambda=0$ belong to two Jordan cells: one of dimension $1$ and one of dimension $2$, \cite[Chapter 3]{horn2012matrix}.

Such a simple example already reveals the issue: either explicitly, as in \eqref{eq:mpi-recurrence-polyhedral}, or implicitly, as in \eqref{eq:mpi_constr_zon}, we manipulate powers of the form
\[
A_\circ^k = (USU^\top)^k = U S^k U^\top.
\]
This becomes numerically unstable whenever $S$ is singular, as is the case here. Therefore, the remainder of the paper exploits the structure of $S$ to perform computations in a well-behaved subspace, before lifting the results back to the original dimension.

\newpage

\section{Main idea}
\label{sec:main_idea}
In the singular case, matrix $A_\circ$ is no longer full rank. Consequently, its image is not full dimensional. As a result, set recurrence \eqref{eq:mpi_rec_standard} may behave erratically or fail altogether, since multiplication by $A_\circ^k$ effectively acts as a projection onto the subspace associated with nonzero poles. In the remainder of this section, we propose constructions that explicitly handle poles at zero, regardless of their nature, for both polyhedral and constrained-zonotope set descriptions.

\subsection{The polyhedral case}

We handle first the case of polyhedral set descriptions \eqref{eq:hrep}--\eqref{eq:mpi-recurrence-polyhedral}, as shown in the following result.

\begin{prop}
\label{prop:p}
Consider the singular closed-loop matrix 
\begin{equation}
\label{eq:schur}
    A_\circ=A+BK = U\left[\begin{tabular}{c|c}
         $S_{11}$& $S_{12}$ \\\hline
         0& $S_{22}$
    \end{tabular}\right]U^\top, 
\end{equation}
whose Schur decomposition has the $(2,2)$ block $S_{22}\in \mathbb{R}^{d_2\times d_2}$ nilpotent of degree $p+1<n$, i.e.,
\begin{equation}
\label{eq:nilpotent}
    \bigl(S_{22}\bigr)^{p+1} = 0, \qquad \text{eig}(S_{22}) = 0.
\end{equation}
The scalars $d_1>0$, $d_2>0$ satisfy $d_1+d_2=n$.

Introducing the shorthand notation
\begin{equation}
    \label{eq:T}
    T=\sum_{i=0}^{p} \bigl(S_{11}\bigr)^{-i}S_{12}\bigl(S_{22}\bigr)^{i},
\end{equation}
the MPI set associated with the state matrix \eqref{eq:schur} and constraints \eqref{eq:hrep} is given by

\begin{equation}
    \label{eq:mpi-schur}
    \Phi_x = \Phi^{[1]}_x \cap \Phi^{[2]}_x, 
\end{equation}

\noindent with
\begin{align}
    \label{eq:mpi-schur_1}
    \Phi_x^{[1]} &= \bigcap\limits_{k=0}^{p-1} \overline{\mathcal X}=\bigcap\limits_{k=0}^{p-1}\biggl\{x:\:FA_\circ^k x\leq \theta\biggr\}, \\
    \label{eq:mpi-schur_2}
    \Phi_x^{[2]} &=\bigcap\limits_{k=p}^{\bar k +p}\overline{\mathcal X}= \biggl\{x:\: F_z \bbm \bigl(S_{11}\bigr)^p &\bigl(S_{11}\bigr)^{p-1}T\ebm U x\leq \theta_z\biggr\}.
    \end{align}
The pair $(F_z, \theta_z)$ defines 
\begin{equation}
        \label{eq:mpi_z}
    \Phi_z\hphantom{^{1}} =\bigcap\limits_{k=0}^{\bar k}\mathcal Z=\{z:\: F_zz\leq \theta_z\},
\end{equation}
the MPI set, computed as in \eqref{eq:mpi_rec_standard}, associated with the dynamics $z^+=S_{11}z$ and constraint set $\mathcal{Z}=\{z:\,F U_1 z\leq\theta\}$. $U_1$ gathers the first $d_1$ columns of $U$ and $U_2$ the remaining $d_2$ columns.
\end{prop}


\begin{proof}
Under the Schur decomposition \eqref{eq:schur}, the dynamics \eqref{eq:LTI_dynamics_closed} and the associated set $\overline{\mathcal X}$ can be expressed in the transformed coordinates $y = U^\top x$, leading to
\begin{subequations}
\begin{align}
    \label{eq:y_dynamics}
    y_{k+1}&=\left[\begin{tabular}{c|c}
         $S_{11}$& $S_{12}$ \\\hline
         0& $S_{22}$
    \end{tabular}\right]y_k, \\
    \label{eq:y_set}\mathcal Y &= \{y:\: FUy\leq \theta\}.
\end{align}
\end{subequations}
Both the dynamics and the associated constraint set follow from the change of variables $x = Uy$ together with the orthogonality property $UU^\top = U^\top U = I$.

Decomposing $y=\bbm y_1^\top & y_2^\top \ebm^\top$ we use \eqref{eq:y_dynamics} to write the evolution of each component of the state
\begin{subequations}
\label{eq:dynamics}
    \begin{align}
        \nonumber y_{1, k+1} &= S_{11}\,y_{1, k} + S_{12}y_{2, k}\\
        &=\bigl(S_{11}\bigr)^{k+1}y_{1,0}+\sum\limits_{i=0}^k \bigl(S_{11}\bigr)^{k-i}S_{12}\bigl(S_{22}\bigr)^{i}y_{2,0},\\
        \nonumber y_{2, k+1} &= S_{22}y_{2, k}\\
        &=\bigl(S_{22}\bigr)^{k+1}y_{2,0}.
    \end{align}
\end{subequations}
Exploiting the nilpotence condition \eqref{eq:nilpotent} we rewrite \eqref{eq:dynamics} into
\begin{subequations}
\label{eq:dynamics_nilpotent}
    \begin{align}
        \nonumber y_{1, k+1} =&  \bigl(S_{11}\bigr)^{k+1}y_{1,0}\\
        &+          \label{eq:dynamics_nilpotent_a}\bigl(S_{11}\bigr)^{k}\sum\limits_{i=0}^{\min(k,\, p)} \bigl(S_{11}\bigr)^{-i}S_{12}\bigl(S_{22}\bigr)^{i}y_{2,0},\\
        \label{eq:dynamics_nilpotent_b}y_{2, k+1} =& \begin{cases}\bigl(S_{22}\bigr)^{k+1}y_{2,0}, & k<p,\\
        0, & k\geq p.
        \end{cases}
    \end{align}
\end{subequations}
For convenience, let us denote $z_k=y_{1,k+p}$ and take $T$ as in \eqref{eq:T}. Consequently, for any $k\geq p$, \eqref{eq:dynamics_nilpotent_a} is equivalently written as (we shift from $k$ to $k-p$)
\begin{equation}
\label{eq:z}
    z_{k+1} = S_{11} z_k, \quad z_0=\bigl(S_{11}\bigr)^py_{1,0}+\bigl(S_{11}\bigr)^{p-1}Ty_{2,0}.
\end{equation}

Using \eqref{eq:dynamics_nilpotent} in \eqref{eq:y_set} we observe a similar behavior for the constraints: for $k\geq p$ only the first component ($y_{1, k}$) is affected: 
\begin{multline}
\label{eq:z_set}
    \mathcal Y=\{FUy_{k}\leq \theta\}=\bigl\{F\bbm U_1 & U_2\ebm \bbm y_{1, k}\\ y_{2, k}\ebm\leq \theta\bigr\}\\
    \overset{\eqref{eq:dynamics_nilpotent}}{=}\{FU_1 y_{1,k}\leq \theta\}\overset{\eqref{eq:z}}{=}\{FU_1 z_{k-p}\leq \theta\}. 
\end{multline}
At this point, we have both ingredients required for the set recurrence \eqref{eq:mpi_rec_standard}: the dynamics given by the first part of \eqref{eq:z} and the constraint set $\mathcal{Z}=\{z:\: FU_1 z\leq \theta\}$. Applying the recurrence yields \eqref{eq:mpi_z}. Next, rewriting $\Phi_z$ using the second part of \eqref{eq:z} allows us to constrain explicitly $y_{1,0}$ and $y_{2,0}$, obtaining \eqref{eq:mpi-schur_2} which enforces the constraints $y_k \in \mathcal Y$ for all $k \geq p$. By additionally imposing the constraints $y_k \in \mathcal Y$ for $k \leq p$, and reversing the change of coordinates $x = Uy$, we obtain \eqref{eq:mpi-schur_1}. Intersecting them gives \eqref{eq:mpi-schur}, thus concluding the proof.
\end{proof}


\subsection{Constrained zonotopic case}

The same line of reasoning may be employed in the constrained zonotopic case \eqref{eq:mpi_constr_zon}--\eqref{eq:mpi_constr_zon}, as illustrated in the next result. 
\begin{prop}
\label{prop:cz}
    Consider the notation from \eqref{eq:schur}--\eqref{eq:T} and the constraint set given in its constrained zonotopic form $\overline{\mathcal X}=\langle c,G,F,\theta\rangle$, defined as in \eqref{eq:x_constr_zon}. Then the associated MPI set is given by 
    \begin{multline}
    \label{eq:mpi-schur-cz}
        \mkern-36mu\Phi_x = \mkern-2mu\biggl\langle \mkern-4mu c^{[1]},\mkern-2mu \bbm G^{[1]} & 0\ebm\mkern-2mu,\mkern-4mu \bbm F^{[1]} & 0\\ 0 & F^{[2]}\\ \bbm \bigl(S_{11}\bigr)^p &\bigl(S_{11}\bigr)^{p-1}T\ebm U^\top G^{[1]} & -G^{[2]}\mkern-4mu\ebm\mkern-4mu,\\
        \bbm \theta^{[1]}\\ \theta^{[2]}\\ c^{[2]}-\bbm \bigl(S_{11}\bigr)^p &\bigl(S_{11}\bigr)^{p-1}T\ebm c^{[1]}\ebm\biggr\rangle
    \end{multline}
    where the ``$[1]$'' parameters describe the constraints imposed until $k<p$:
    \begin{equation}
    \label{eq:omega_x_1_cz}
        \Phi_x^{[1]}=\bigcap\limits_{k=0}^{p-1} \overline{\mathcal X} = \bigl\langle c^{[1]}, G^{[1]}, F^{[1]}, \theta^{[1]}\bigr\rangle
    \end{equation}
and the ``$[2]$'' parameters describe the constraints applied to $k\geq p$ and, under dynamics \eqref{eq:z} and set recurrence \eqref{eq:mpi_constr_zon} which stops after $\bar k$ iterations \cite{10552328}:
    \begin{multline}
        \label{eq:omega_x_2_cz}
        \Phi_x^{[2]}=\bigcap\limits_{k=p}^{\bar k +p} {\mathcal X}\\
        =\{x:\: \bbm \bigl(S_{11}\bigr)^p &\bigl(S_{11}\bigr)^{p-1}T\ebm U^\top x\in \Phi_z\}, \theta^{[2]}\bigr\rangle.
    \end{multline}
    where
    \begin{equation}
        \label{eq:mpi_z_cz}
        \Phi_z =\bigcap\limits_{k=0}^{\bar k} {\mathcal Z}=\langle c^{[2]}, G^{[2]}, F^{[2]}, \theta^{[2]}\rangle,
    \end{equation}
    is the MPI set associated with the dynamics $z^+=S_{11}z$ and constraint set 
    \begin{equation}
    \label{eq:z-cz}
        \mathcal Z = \biggl\langle U_1^\top c, U_1^\top G, \bbm F\\ U_2^\top G\ebm, \bbm \theta\\ -U_2^\top c\ebm\biggr\rangle.
    \end{equation}
    All $(\cdot)^{[1]},(\cdot)^{[2]}$ parameters are obtained recursively as in \eqref{eq:mpi_constr_zon}.
\end{prop}
\begin{proof}
Using recurrence \eqref{eq:mpi_constr_zon} until step $k=p-1$ for closed-loop matrix \eqref{eq:schur} and constraint set \eqref{eq:x_constr_zon} gives $\Phi_x^{[1]}$ as in \eqref{eq:omega_x_1_cz}. Next, for $k\geq p$, we consider the same change of coordinates as before, $y=U^\top x$. Using definition \eqref{eq:con_poly_zono}, we have that $x_{k+1}\in \overline{\mathcal X}$ can be written as
\begin{subequations}
\label{eq:y_cz}
\begin{align}
\label{eq:y_cz_1}
    y_{1, k+1}&=U_1^\top c + U_1^\top G\lambda_{k+1},\quad F_1\lambda_{k+1} = \theta_1,\\
\label{eq:y_cz_2}
y_{2, k+1}&=U_2^\top c + U_2^\top G\lambda_{k+1},\quad F_2\lambda_{k+1} = \theta_2.
\end{align}    
\end{subequations}
Due to \eqref{eq:nilpotent} and as per \eqref{eq:dynamics_nilpotent_b} we have that $y_{2, k+1}=0$. This acts as an additional equality constraint on $\lambda_{k+1}$ in \eqref{eq:y_cz_2}. Introducing it into \eqref{eq:y_cz_1}, leads to inclusion $y_{1,k+1}\in \mathcal Z$, with $\mathcal Z$ defined as in \eqref{eq:z-cz}. With notation \eqref{eq:z} and the set \eqref{eq:z-cz}, set recurrence \eqref{eq:mpi_constr_zon} leads to \eqref{eq:mpi_z_cz}. Using the second part of \eqref{eq:z} and \eqref{eq:omega_x_1_cz} and \eqref{eq:omega_x_2_cz} leads to \eqref{eq:mpi-schur-cz}, thus concluding the proof.
\end{proof}

\subsection{Analysis}

Comparing Propositions~\ref{prop:p} and~\ref{prop:cz}, we observe that the same set names appear in both statements. This is not a mistake but a deliberate choice: the underlying geometric objects are identical. Indeed, constrained zonotopes and polyhedra have been shown to be equivalent set representations~\cite{raghuraman2022set}. The differences arise only in the specific algorithms used to manipulate them (in this case, the operations of matrix multiplication and set intersection). To emphasize this point, we present Algorithm~\eqref{alg:MPI_poles_in_zero}, which enumerates the computation steps for both polyhedral and constrained-zonotope representations.

\begin{algorithm}[!ht]
\SetKwComment{Comment}{}{}
\KwIn{$A_{\circ}$, $\overline{\mathcal{X}}$ as in \eqref{eq:hrep} or \eqref{eq:con_poly_zono}}
\KwOut{the MPI set $\Phi_x$}

\label{step:1}Compute the sorted real Schur form: $A_{\circ} = USU^T$ with $S$ and $p+1$ as in \eqref{eq:schur}
\BlankLine


\textbf{Select procedure:\\ 
\hphantom{poly}
Polyhedral \hspace{2.56em} \vrule  \hspace{3em} Constrained Zonotope}
\BlankLine

$\mathcal{Z}$ from \eqref{eq:z_set} \; \hspace{4.25em} \vrule  \hspace{3.25em} $\,$ 
$\mathcal{Z}$ from \eqref{eq:z-cz}
\BlankLine

Take $A_\circ = S_{11}$, $\overline{\mathcal{X}}=\mathcal{Z}$ and compute reduced MPI set
\BlankLine

\hspace{3.77em}$\Phi_z$ as in \eqref{eq:mpi_z}\; \hspace{4.25em} \vrule   \hspace{3.5em} $\,$ $\Phi_z$ as in \eqref{eq:mpi_z_cz}
\BlankLine

Compute auxiliary sets:
\BlankLine

\hspace{3.5em} $\Phi_{x}^{[1]}$ from \eqref{eq:mpi-schur_1} \; \hspace{3.5em} \vrule   \hspace{3.5em} $\,$ 
$\Phi_{x}^{[1]}$ from \eqref{eq:omega_x_1_cz}
\BlankLine

\hspace{3.5em} Compute $T$ as in \eqref{eq:T}
\BlankLine

\hspace{3.5em} $\Phi_{x}^{[2]}$ from \eqref{eq:mpi-schur_2} \; \hspace{3.5em} \vrule   \hspace{3.5em} $\,$ 
$\Phi_{x}^{[2]}$ from \eqref{eq:omega_x_2_cz}
\BlankLine

Compute MPI set $\Phi_{x} = \Phi_{x}^{[1]} \bigcap \Phi_{x}^{[2]}$
\caption{MPI set for singular cases }
\label{alg:MPI_poles_in_zero}
\end{algorithm}
Noteworthy, step~1 of the algorithm implements the \emph{ordered Schur} decomposition from \cite{brandts2002matlab}, modified such that it arranges the Jordan and Schur blocks in descending order from largest to smallest, in magnitude. Thus, all the blocks with zero eigenvalue (if any) accumulate in the $(2,2)$ block, as shown in \eqref{eq:schur}.

\noindent Several remarks are in order. 
\begin{rem}
The set recurrence \eqref{eq:mpi_rec_standard} terminates at some finite index $\bar k$. Although unlikely, it is possible that $\bar k < p$. In this case, the construction in \eqref{eq:mpi-schur} or \eqref{eq:mpi-schur-cz} terminates without the need to compute \eqref{eq:mpi_z} or \eqref{eq:mpi_z_cz} or, indeed, the parts of \eqref{eq:mpi-schur_1}, \eqref{eq:omega_x_1_cz} which correspond to indices $\bar k+1, \ldots, p-1$. \eor
\end{rem}
\begin{rem}
Formulations \eqref{eq:mpi-schur} and \eqref{eq:mpi-schur-cz} also provide geometric insight into why a singular matrix disrupts the standard recurrence \eqref{eq:mpi_rec_standard}. The set $\Phi_z$ is computed in $\mathbb{R}^{d_1}$. When embedded into $\mathbb{R}^n$ as $\Phi_x^{[2]}$ after the change of variable $y=U^\top x$, it becomes unbounded along the subspace $\mathbb{R}^{d_2}$. Consequently, the bounds that ensure a well-behaved set must arise explicitly from the first $p$ iterations, as collected in \eqref{eq:mpi-schur_1} and \eqref{eq:omega_x_1_cz}, respectively.\eor
\end{rem}


\begin{rem}
Index $p+1$, as defined in \eqref{eq:nilpotent}, corresponds to the dimension of the largest Jordan cell of $S_{22}$. Inspecting the Jordan indices through the Schur form of $S_{22}$, \cite{kublanovskaya1982construction}, provides additional insight into the nature of the constraints that arise at a given time instant $k$. 

For illustration, suppose that the block dimensions belong to the (not necessarily unique) ordered set $\{p_1 \leq p_2 \leq \ldots \leq p_m\}$, with multiplicities $\{j_1, j_2, \ldots, j_m\}$. Then, for indices satisfying $p_i < k \leq p_{i+1}$, the constraints applied at step $k$ (i.e., $A^k x \in \overline{\mathcal X}$) describe a geometric object embedded in a subspace of dimension $n-\sum_{\ell=1}^{i} p_\ell j_\ell$. Once $k \geq p = \max_{\ell=1,\ldots,m} p_\ell$, the corresponding subspace dimension reaches $d_1 = n-d_2 = n-\sum_{\ell=1}^{m} p_\ell j_\ell$.
    \eor
\end{rem}





\section{Illustrative example}
\label{sec:ill_example}

The following results were obtained using MPT3~\cite{MPT3} and CasADi~\cite{andersson2019casadi}. All computations were performed on a system running Ubuntu~22.04, equipped with an AMD Ryzen~9~3900X 3.8~GHz CPU (12 cores, 24 threads) and 48\,GB of RAM.

The MATLAB implementation corresponding to Prop.~\ref{prop:p} is available in the subdirectory {/replan-public/maximal-positive-invariant-set/LCSS-2026} of the GitLab repository:
\url{https://gitlab.com/replan/replan-public.git}.

\subsection{Numerical example in $\mathbb R^6$}

We recall the numerical example from \secref{sec:justication}. To this we attach state and input constraints $\mathcal X = \{x\in \mathbb R^{6}:\:\|x\|_{\infty} \leq 1\},\:  \mathcal U=\{u\in \mathbb R^{2}:\:\|u\|_{\infty} \leq 1\}$ which correspond to
\begin{equation*}
    \overline{\mathcal X}= \biggl\{\bbm I_6 \\ I_2K\ebm x\leq \bbm \mathbf{1}_6\\ \mathbf{1}_2\ebm\biggr\},
\end{equation*}
given as in \eqref{eq:barx}. Since in $S$, the upper-triangular matrix for the Schur factorization of $A_\circ$, we have two Jordan blocks (of dimensions $p_1=1$ and $p_2=2$), we have that the rank deficiency is $d_2 = p_1+p_2=3$ and $p=\max(p_1,p_2)=2$. Consequently, $d_1= n-d_2 = 3$ and $S_{11}$ is the $3\times 3$ matrix 
\begin{equation*}
    S_{11}=\begin{bmatrix}
0.70 & -0.87 & 0.49 \\
\hphantom{0.0}0 & \hphantom{-}0.50 & 0.12 \\
\hphantom{0.0}0 & \hphantom{-0.0}0 & 0.20
\end{bmatrix}.
\end{equation*}
Correspondingly, we compute the set $\mathcal Z\subset \mathbb R^3$ as shown below \eqref{eq:mpi_z}. Having both these elements we compute the reduced MPI set $\Phi_z$ from \eqref{eq:mpi_z} in $\overline{k}= 3$ iterations, as  depicted in \figref{fig:example_LTI_6dim}.

\begin{figure}[!ht]
\centering
\includegraphics[width=0.5\textwidth]{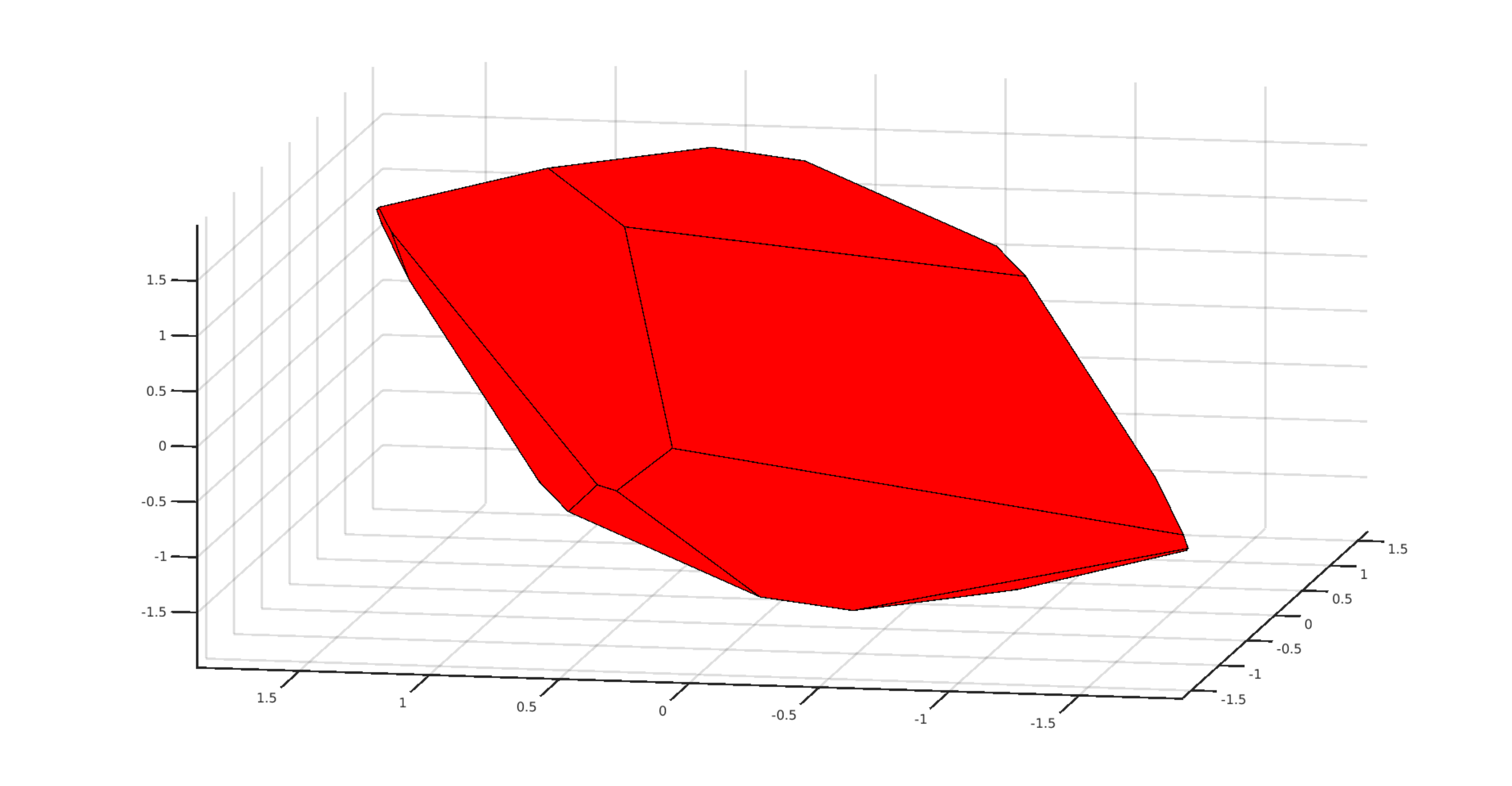}
\caption{MPI set $\Phi_z$ in $\mathbb R^3$.}
\label{fig:example_LTI_6dim}
\end{figure}

By selecting the blocks $S_{12}$ and $S_{22}$ from $S$ we compute, as in \eqref{eq:T}: 
\[
T = \begin{bmatrix}
-0.30 & 33.06 & -11.76 \\
-0.22 & \hphantom{0}2.23 & -\hphantom{0}0.82 \\
-0.39 & \hphantom{0}0.76 & -\hphantom{0}0.15
\end{bmatrix}
\]
which is used, as in \eqref{eq:mpi-schur_2} to obtain $\Phi_x^{[2]}\subset \mathbb R^6$. Set $\Phi_x^{[1]}\subset \mathbb R^6$ is obtained with $p=3$ iterations as in \eqref{eq:mpi-schur_1}. Intersecting the two, we obtain the desired result, the MPI set $\Phi_x$, characterizing the original dynamics and constraints. 
Total computation time for all operations was $0.4$ seconds.


\newpage

\subsection{The CSE example \cite{Leibfritz2006COMPleibCM}}
The Coupled Spring Experiment (CSE) considers a variable number of masses interconnected through springs and dampers. The system state consists of the positions and velocities of all masses, while the inputs are two external forces applied at the ends of the chain. The continuous-time model is given by
\begin{align}
\label{eq:cse1}
    \dot x &= 
    \underbrace{\begin{bmatrix}
    0 & I \\
    -M_{c}^{-1}K_{c} & -M_{c}^{-1}L_{c}
    \end{bmatrix}}_{\text{A}\in \mathbb R^{2\ell \times 2\ell}}x + 
    \underbrace{\begin{bmatrix}
    0 \\
    M_{c}^{-1}D_{c}
    \end{bmatrix}}_{\text{B}\in \mathbb R^{2\ell \times 2}} u,
\end{align}

where $M_{c} = \mu I$, $L_{c} = \delta I$, 
\begin{align}
    K_{c} = k\begin{bmatrix}
    1 & -1 & \cdots & 0 & 0 \\
    -1 & -2 & \ddots & 0 & 0 \\
    \vdots & \ddots & \ddots & \ddots & \vdots \\
    0 & 0 & \ddots & -2 & -1 \\
    0 & 0 & \cdots & -1 & 1
    \end{bmatrix} \nonumber, \; \text{and} \; D_{c} = \begin{bmatrix}
    1 & \hphantom{-}0 \\
    0 & \hphantom{-}0 \\
    \vdots & \hphantom{-}\vdots \\
    0 & \hphantom{-}0 \\
    0 & -1
    \end{bmatrix}, \nonumber
\end{align}
with $\mu = 4, \; \delta = 1, \; k = 1$. The system is discretized using forward Euler method with a sampling time of $1$ sec. The state and input constraints are: $\mathcal X = \{x\in \mathbb R^{2\ell}:\:\|x\|_{\infty} \leq 1\},\:  \mathcal U=\{u\in \mathbb R^{2}:\:\|u\|_{\infty} \leq 1\}$. By varying the number of masses through the parameter $\ell$, the state dimension becomes $2\ell$.

The application of Algorithm~\ref{alg:MPI_poles_in_zero} in the polyhedral branch (Prop.~\ref{prop:p}) to the CSE dynamics, closed using both pole placement (all eigenvalues  in the unit disk) and a Riccati-based method (yielding a rank-deficient matrix), is illustrated in Fig.~\ref{fig:plot_data_simulations}. 

\begin{figure}[!ht]
\centering
\includegraphics[width=\columnwidth]{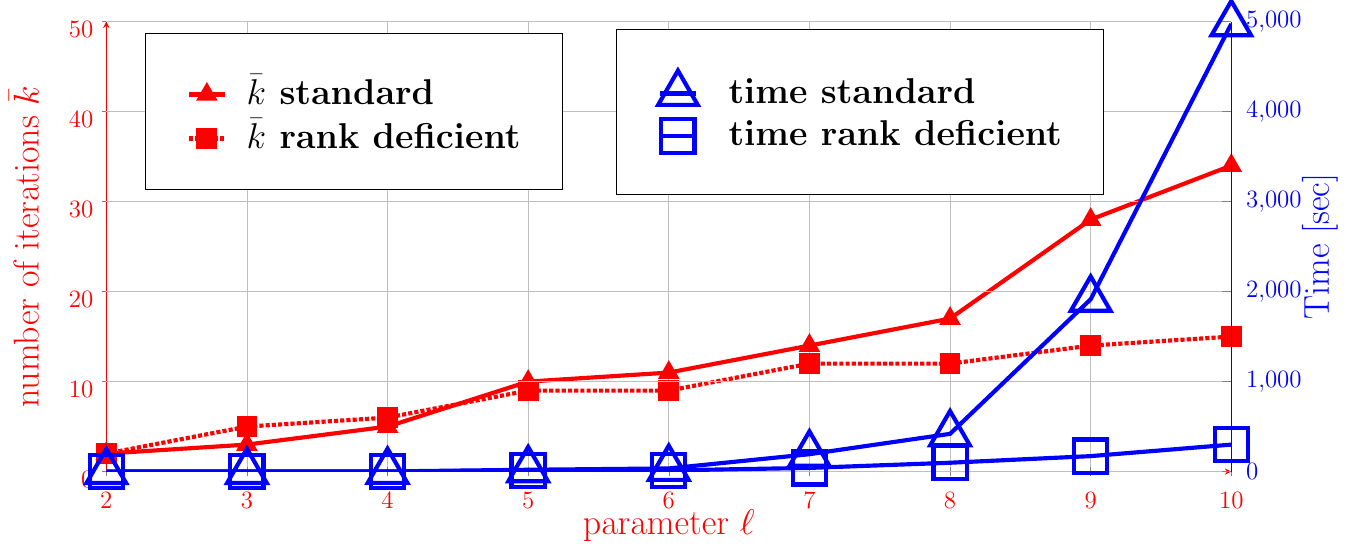}
\caption{Plot number of iterations and commutation time.}
\label{fig:plot_data_simulations}
\end{figure}

We observe that:
\begin{itemize}
    \item Starting from higher dimensions ($\ell = 8$, corresponding to $x \in \mathbb{R}^{16}$), the number of iterations required by the standard algorithm ($\bar k$) increases significantly, whereas the increase is moderate for the rank-deficient method;
    \item The computation time further highlights the effectiveness of the proposed approach: the standard method becomes considerably more expensive for $\ell \geq 8$, while the rank-deficient method remains computationally efficient;
    \item The number of inequalities $\bar q$ required to represent the resulting polyhedral MPI sets is also larger for the standard approach (starting from $\ell = 5$). For instance, for $\ell = 10$, the standard method yields $\bar q = 1496$ inequalities, whereas the rank-deficient method reduces to $\bar q = 704$.
\end{itemize}

The resulting MPIs sets are not geometrically the same. This outcome is expected because of the pole placement, which changes the initial set $\overline{\mathcal{X}}$ (as shown in Fig.~\ref{fig:example_LTI_MPI}).


\section{Conclusions}
\label{sec:conclusions}


We have presented a robust algorithm for computing the MPI set in the presence of rank-deficient dynamics, using the Schur factorization. The approach accommodates multiple zero eigenvalues with arbitrary multiplicities, and provides the resulting MPI set in both polyhedral and constrained-zonotope representations. Future work will extend this framework to the more general Maximal Output Admissible Set (MOAS), where we anticipate employing an SVD-based factorization.




\begin{thebibliography}{10}
\providecommand{\url}[1]{#1}
\csname url@samestyle\endcsname
\providecommand{\newblock}{\relax}
\providecommand{\bibinfo}[2]{#2}
\providecommand{\BIBentrySTDinterwordspacing}{\spaceskip=0pt\relax}
\providecommand{\BIBentryALTinterwordstretchfactor}{4}
\providecommand{\BIBentryALTinterwordspacing}{\spaceskip=\fontdimen2\font plus
\BIBentryALTinterwordstretchfactor\fontdimen3\font minus \fontdimen4\font\relax}
\providecommand{\BIBforeignlanguage}[2]{{%
\expandafter\ifx\csname l@#1\endcsname\relax
\typeout{** WARNING: IEEEtran.bst: No hyphenation pattern has been}%
\typeout{** loaded for the language `#1'. Using the pattern for}%
\typeout{** the default language instead.}%
\else
\language=\csname l@#1\endcsname
\fi
#2}}
\providecommand{\BIBdecl}{\relax}
\BIBdecl

\bibitem{ossareh_complexity_2024}
H.~R. Ossareh and I.~Kolmanovsky, ``\BIBforeignlanguage{en}{On {Complexity} {Bounds} for the {Maximal} {Admissible} {Set} of {Linear} {Time}-{Invariant} {Systems}},'' \emph{\BIBforeignlanguage{en}{IEEE Transactions on Automatic Control}}, pp. 1--8, 2024.

\bibitem{marpi2024lcss}
A.~Dey and S.~Bhasin, ``Computation of maximal admissible robust positive invariant sets for linear systems with parametric and additive uncertainties,'' \emph{IEEE Control Systems Letters}, vol.~8, pp. 1775--1780, 2024.

\bibitem{chen1998quasi}
H.~Chen and F.~Allgower, ``A quasi-infinite horizon nonlinear model predictive control scheme with guaranteed stability,'' \emph{Automatica}, vol.~34, no.~10, pp. 1205--1217, 1998.

\bibitem{rakovic2018handbook}
S.~V. Rakovi{\'c} and W.~S. Levine, \emph{Handbook of model predictive control}.\hskip 1em plus 0.5em minus 0.4em\relax Springer, 2018.

\bibitem{10552328}
B.~Gheorghe, D.-M. Ioan, F.~Stoican, and I.~Prodan, ``Computing the maximal positive invariant set for the constrained zonotopic case,'' \emph{IEEE Control Systems Letters}, vol.~8, 2024.

\bibitem{gilbert1991linear}
E.~Gilbert and K.~Tan, ``{Linear systems with state and control constraints: the theory and application of maximal output admissible sets},'' \emph{IEEE Tran. on Automatic Control}, vol.~36, no.~9, pp. 1008--1020, 1991.

\bibitem{kolmanovsky_theory_nodate}
I.~Kolmanovsky and E.~G. Gilbert, ``{Theory} and computation of disturbance invariant sets for discrete-time linear systems,'' \emph{Mathematical Problems in Engineering}, vol.~4, no.~4, pp. 317--363, 1998.

\bibitem{rakovic2022implicit}
S.~V. Rakovi{\'c} and S.~Zhang, ``The implicit maximal positively invariant set,'' \emph{IEEE Tran. on Automatic Control}, vol.~68, no.~8, pp. 4738--4753, 2023.

\bibitem{blanchini2008set}
F.~Blanchini and S.~Miani, \emph{Set-theoretic methods in control}.\hskip 1em plus 0.5em minus 0.4em\relax Springer, 2008, vol.~78.

\bibitem{houska_polyhedral_2025}
B.~Houska, M.~A. Müller, and M.~E. Villanueva, ``\BIBforeignlanguage{en}{Polyhedral control design: {Theory} and methods},'' \emph{\BIBforeignlanguage{en}{Annual Reviews in Control}}, vol.~60, p. 100992, 2025.

\bibitem{ziegler2012lectures}
G.~M. Ziegler, \emph{Lectures on polytopes}.\hskip 1em plus 0.5em minus 0.4em\relax Springer Science \& Business Media, 2012, vol. 152.

\bibitem{scott2016constrained}
J.~K. Scott, D.~M. Raimondo, G.~R. Marseglia, and R.~D. Braatz, ``Constrained zonotopes: A new tool for set-based estimation and fault detection,'' \emph{Automatica}, vol.~69, pp. 126--136, 2016.

\bibitem{raghuraman2022set}
V.~Raghuraman and J.~P. Koeln, ``Set operations and order reductions for constrained zonotopes,'' \emph{Automatica}, vol. 139, p. 110204, 2022.

\bibitem{ionescu1997general}
V.~Ionescu, C.~Oara, and M.~Weiss, ``General matrix pencil techniques for the solution of algebraic riccati equations: a unified approach,'' \emph{IEEE Transactions on Automatic Control}, vol.~42, no.~8, pp. 1085--1097, 1997.

\bibitem{van1984deadbeat}
P.~Van~Dooren, ``Deadbeat control: A special inverse eigenvalue problem,'' \emph{BIT Numerical Mathematics}, vol.~24, no.~4, pp. 681--699, 1984.

\bibitem{horn2012matrix}
R.~A. Horn and C.~R. Johnson, \emph{Matrix analysis}.\hskip 1em plus 0.5em minus 0.4em\relax Cambridge university press, 2012.

\bibitem{brandts2002matlab}
J.~H. Brandts, ``Matlab code for sorting real {Schur} forms,'' \emph{Numerical linear algebra with applications}, vol.~9, no.~3, pp. 249--261, 2002.

\bibitem{kublanovskaya1982construction}
V.~Kublanovskaya, ``Construction of a canonic basis for matrices and pencils of matrices,'' \emph{Journal of Soviet Mathematics}, vol.~20, no.~2, pp. 1929--1942, 1982.

\bibitem{MPT3}
M.~Herceg, M.~Kvasnica, C.~Jones, and M.~Morari, ``{Multi-Parametric Toolbox 3.0},'' in \emph{Proc.~of the European Control Conference}, Z\"urich, Switzerland, July 17--19 2013, pp. 502--510.

\bibitem{andersson2019casadi}
J.~A. Andersson, J.~Gillis, G.~Horn, J.~B. Rawlings, and M.~Diehl, ``{CasADi}: a software framework for nonlinear optimization and optimal control,'' \emph{Mathematical Programming Computation}, vol.~11, no.~1, pp. 1--36, 2019.

\bibitem{Leibfritz2006COMPleibCM}
F.~Leibfritz, ``{COMPleib: COnstrained Matrix–optimization Problem library} – a collection of test examples for nonlinear semidefinite programs, control system design and related problems.''\hskip 1em plus 0.5em minus 0.4em\relax Department of Mathematics, Univ. Trier, Germany, Tech. Rep., 2006.

\end{thebibliography}
\end{document}